\documentclass[twocolumn,10pt]{article}

\usepackage{usenix}
\usepackage[utf8]{inputenc}

\usepackage{amsmath}
\usepackage{amssymb}
\usepackage{amsthm}
\usepackage{hyperref}
\usepackage[capitalise]{cleveref}
\usepackage{csquotes}
\usepackage{mathtools}
\usepackage{nicefrac}
\usepackage{algorithm}
\usepackage{algpseudocode}
\usepackage{tikz}
\usepackage{todonotes}
\usepackage{booktabs}
\usepackage{siunitx}
\usepackage{qtree}
\usepackage{fontawesome}

\usetikzlibrary{shapes}
\usetikzlibrary{calc}

\usepackage[backend=biber,maxbibnames=1]{biblatex}
\newcommand*\prot{FoSAM}
\newcommand*\etal{\textit{et al.}}

\newcommand*\integers{\mathbb{Z}}
\newcommand*\bits{\left\{0, 1\right\}}
\newcommand*\unidraw{\leftarrow^{\mathrm{R}}}
\newcommand*\Group{\mathbb{G}}
\newcommand*\Ge{\mathbb{G}_1}
\newcommand*\Gz{\mathbb{G}_2}
\newcommand*\Gt{\mathbb{G}_\tau}

\newcommand*\omnet{OMNeT\nolinebreak[4]\hspace{-0.05em}\raisebox{.12ex}{\bf ++}}

\newcommand*\FsKeyGen{\mathrm{FS.KeyGen}}
\newcommand*\FsUpdate{\mathrm{FS.Update}}
\newcommand*\FsEncrypt{\mathrm{FS.Encrypt}}
\newcommand*\FsDecrypt{\mathrm{FS.Decrypt}}
\newcommand*\PubKey{\mathrm{pk}}
\newcommand*\PrivKey{\mathrm{sk}}
\newcommand*\FsBreakin{\mathtt{breakin}}
\newcommand*\FsChallenge{\mathtt{challenge}}

\newcommand*\HibeKeyGen{\mathrm{HIBE.KeyGen}}
\newcommand*\HibeUskGen{\mathrm{HIBE.SkGen}}
\newcommand*\HibeEnc{\mathrm{HIBE.Enc}}
\newcommand*\HibeDel{\mathrm{HIBE.Del}}
\newcommand*\HibeDec{\mathrm{HIBE.Dec}}
\newcommand*\HibeId{\mathtt{id}}
\newcommand*\MasterSecret{\mathrm{msk}}
\newcommand*\DelKey{\mathrm{dk}}

\newcommand*\adversary{\mathcal{A}}
\newcommand*\challenger{\mathcal{C}}
\newcommand*\protocol{\Pi}
\newcommand*\notion{X}
\newcommand*\RML{RM\bar{L}}

\newcommand*\Input{\textbf{input}\ }

\newtheorem{definition}{Definition}
\newtheorem{theorem}{Theorem}
\newtheorem{lemma}{Lemma}
\newtheorem{game}{Game}

\DeclarePairedDelimiter\abs{\lvert}{\rvert}

\title{\prot: Forward Secret Messaging in Ad-Hoc Networks}
\author{%
    {\rm Daniel Schadt}\\
    Karlsruhe Institute of Technology
    \and
    {\rm Christoph Coijanovic}\\
    Karlsruhe Institute of Technology
    \and
    {\rm Thorsten Strufe}\\
    Karlsruhe Institute of Technology
}
\date{}

\begin{document}

\maketitle

\begin{abstract}
Apps such as Firechat and Bridgefy have been used during recent protests in Hong Kong and Iran, as they allow communication over ad-hoc wireless networks even when internet access is restricted.
However, these apps do not provide sufficient protection as they do not achieve forward secrecy in unreliable networks.
Without forward secrecy, caught protesters' devices will disclose all previous messages to the authorities, putting them and others at great risk.
In this paper, we introduce FoSAM, the first protocol to provide proven anonymous and forward secret messaging in unreliable ad-hoc networks.
Communication in FoSAM requires only the receiver's public key, rather than an interactive handshake.
We evaluate the performance of FoSAM using a large-scale simulation with different user movement patterns, showing that it achieves between 92\% and 99\% successful message delivery.
We additionally implement a FoSAM prototype for Android.

\end{abstract}

\section{Introduction}

Instant messaging is a valuable tool for organizing large-scale protests.
However, authoritarian governments may restrict internet connectivity in order to hinder the communication between protesters, as has happened in Russia\footnote{Evidence of internet disruptions in Russia during Moscow opposition protests, \emph{NetBlocks}, 2019} and Iran\footnote{Iran's Internet Shutdown Hides a Deadly Crackdown, \emph{Wired}, 2022}.
This renders even secure chat infrastructures like Signal or Matrix useless, and as a result, apps like Firechat and Bridgefy that rely on peer-to-peer wireless ad-hoc networks (MANETs) and function without the internet have been used by protesters.

Such peer-to-peer apps use available low-range wireless communication media such as \emph{Bluetooth}, \emph{Bluetooth Low Energy} or \emph{Wi-Fi Direct} to pass messages from device to device.
Messages are flooded through the network to enable communication between users who do not have a direct wireless connection between their devices.
As long as there is a path of devices exchanging messages between two users, they can communicate. 
However, due to the dynamic nature of the network and the movement of users however, the message transport is unreliable, and messages are only delivered on a best effort basis.

Security in a MANET messaging system is essential, as adversaries can easily infiltrate the network with malicious devices and eavesdrop on nearby communications.
To protect the privacy of protesters, the protocol must therefore achieve receiver-message unlinkability~\cite{kuhn2019}, that is, it must hide the intended receiver of a message even against malicious devices.
Further, key compromise is a real threat when government agents can physically confiscate protesters' devices to read out the secret keys.
To protect against this, the protocol should achieve forward secrecy~\cite{guenther1990}, meaning that even if a key is compromised, it cannot be used to decrypt messages captured in the past.
In addition, users should be able to add new contacts without a reliable internet connection.

Typically, cryptographic techniques such as the Diffie-Hellman key exchange~\cite{diffie1976} or Signal's double ratchet~\cite{alwen2019} are implemented to achieve forward secrecy.
However, their implementation is very challenging in the MANET setting because the network is unreliable and messages can be lost, arrive out of order, or be delayed.
As a result, existing systems do not meet these stringent security requirements~\cite{albrecht2022}, putting protesters at risk.

To fill this niche, we design \prot{}, a system that is both private and achieves forward secrecy in unreliable networks.
We do this by basing forward secrecy not on the chronological delivery of messages, but rather on timed epochs, so that even lost or out-of-order messages are protected.
Once an epoch has passed, messages sent during a previous epoch cannot be decrypted.

Our encryption scheme is based on the works of Canetti \etal{}~\cite{canetti2003} and Blazy \etal{}~\cite{blazy2014} and fulfills our security goals:
First, our scheme provides key-privacy, ensuring that messages reveal no information about their intended recipient.
Second, we achieve forward secrecy without interaction, such that messages are protected even when they get lost.
Third, we can add contacts simply by adding their public key, requiring no handshake to start a conversation.
This means that for example protesters can add new contacts via a QR code from a poster or a flyer, even after the protest has started.

As our forward secrecy is based on time epochs, it is important that the devices' clocks are synchronized.
We do not require a strict millisecond accuracy, but the more drift there is between two devices, the more likely it is that messages will be encrypted \enquote{for the wrong epoch}.
To avoid this issue, we allow devices to synchronize their clocks with their neighbors using an adapted version of the leader election algorithm of Byrenheid \etal{}~\cite{byrenheid2020}.
In our evaluation, we show that in most scenarios, the time can be synchronized in less than 10 minutes, and more than 90\% of messages arrive successfully.

In summary, we make the following contributions:

\begin{itemize}
    \item
    We adapt the leader election algorithm of Byrenheid \etal{}~\cite{byrenheid2020} to do time synchronization and evaluate its performance.
    
    \item
    We devise a messaging protocol for peer-to-peer wireless networks and prove it to be forward secret.

    \item
    We evaluate the expected performance via simulations and benchmarks.

    \item
    We provide a functioning prototype application that implements our protocol and can be used on Android smartphones.
\end{itemize}

\section{Related Work}

A few systems aim to provide messaging in ad-hoc networks with various trade-offs.

Rangzen~\cite{lerner2016} provides a messaging system that allows its users to \enquote{anonymously get the word out}.
It focuses on providing anonymous broadcasting in scenarios where the internet is unavailable, by having devices opportunistically forward messages they encounter.
It also uses social graphs to assign priorities to messages and prevent denial of service attacks.
Rangzen however does not provide forward secrecy, as all messages are publicly visible.
Because messages are publicly accessible, an adversary can also single out a user and determine what the user is sending.

Moby~\cite{pradeep2022} is another system that provides messaging during an internet blackout and explicitly implements forward secrecy.
Their message distribution works similarly to Rangzen (opportunistic exchange of messages), and they use the double ratchet~\cite{perrin2016} to ensure forward secrecy.
However, their design targets long-term messaging, that is messages that arrive over the span of several days or weeks.

Another system has been proposed by Perry \etal{}~\cite{perry2022}.
Their system is aimed at real-time communication in large gatherings, but their encryption does not provide forward secrecy.

In addition to academic works, some systems have also been deployed in practice:

Firechat is an app that provided messaging in ad-hoc networks, and it has been used in various protests\footnote{FireChat --- the messaging app that’s powering the Hong Kong protests, \emph{The Guardian}, 2014}\footnote{Firechat updates as 40,000 Iraqis download `mesh' chat app in censored Baghdad, \emph{The Guardian}, 2014}.
The app supported end-to-end encrypted chats, but it is unclear what security guarantees it provided.
In addition, the app has been discontinued and is no longer available for download.

Bridgefy~\footnote{\url{https://bridgefy.me/}} is similar to Firechat.
While they have adopted Signal's encryption for end-to-end encrypted communication, Bridgefy has been shown to be insecure~\cite{albrecht2022}.
Albrecht \etal{} demonstrate that messages contain enough metadata that an adversary can infer social graphs, and they were able to eavesdrop on one-to-one messages due to a vulnerability in the implementation.

Briar~\footnote{\url{https://briarproject.org/}} is a messenger that relays messages over Tor (when the internet is available) or over Bluetooth and WiFi (when the internet is unavailable).
However, unlike Firechat and Bridgefy, it only relays messages over direct connections.
Relaying a message over multiple hops only works if the message is a group message sent to all users along the path.
As such, the range of private messages is limited.

We therefore fill the gap of forward-secret real-time communication in an ad-hoc network by designing a \enquote{secure Firechat}.

\section{Background}

Throughout this paper, we will use \(\Ge\), \(\Gz\) and \(\Gt\) to be cyclic groups with generators \(g_1\), \(g_2\) and \(g_\tau\).
Further, let \(e : \Ge \times \Gz \rightarrow \Gt\) be an efficiently computable bilinear map between them.
A bilinear map is a map with the following properties~\cite{boneh2005}:

\begin{itemize}
    \item
    Bilinearity: for all \(u \in \Ge, v \in \Gz\) and \(a, b \in \integers\), we have \(e(u^a, v^b) = e(u, v)^{ab}\).

    \item
    Non-degeneracy: \(e(g_1, g_2) \ne 1\).
\end{itemize}

We base the security of our construction on the \emph{Matrix Diffie-Hellman Assumption} (MDDH), a generalization of the Diffie-Hellman Assumption that roughly states that you cannot efficiently determine whether \enquote{a given vector in \(\Group^l\) is spanned by the columns of a certain matrix \([A] \in \Group^{l \times k}\)}~\cite{escala2017}.

\subsection{Forward Secrecy}

Forward secrecy intuitively states that a key compromise at a certain time does not compromise ciphertexts from before the key compromise.
This is achieved by \emph{evolving} the key, which is only possible in one direction.
Formally, we use the definitions of Canetti \etal{} to capture the notion of key evolution and forward secrecy:

\begin{definition}[ke-PKE]
    A key-evolving public key encryption scheme consists of four algorithms \((\FsKeyGen, \FsUpdate, \FsEncrypt, \FsDecrypt)\), with the following properties~\cite{canetti2003}:
    
    \begin{itemize}
        \item
        \(\FsKeyGen(\lambda) \to (\PubKey, \PrivKey_0)\) takes as input the security parameter and returns a public-key/secret-key pair \((\PubKey, \PrivKey_0)\).
    
        \item
        \(\FsUpdate(\PrivKey_i, i)\to \PrivKey_{i+1}\) takes as input a secret key \(\PrivKey_i\) and the epoch \(i\), and returns as output the secret key of the next epoch \(\PrivKey_{i + 1}\).
    
        \item
        \(\FsEncrypt(m, i, \PubKey)\to c\) takes as input a message \(m\), the current epoch \(i\) and the public key \(\PubKey\) and returns the ciphertext \(c\).
    
        \item
        \(\FsDecrypt(c, i, \PrivKey_i)\to m\) takes as input the ciphertext \(c\) for epoch \(i\) and the secret key for the right epoch \(\PrivKey_i\), and returns the message \(m\).
    \end{itemize}
    
    We require that the encryption scheme is correct, meaning that \(\FsDecrypt(\FsEncrypt(m, i, \PubKey), \PrivKey_i) = m\) for all messages \(m\) and all \((\PubKey, \PrivKey_i) \leftarrow \FsKeyGen\).
\end{definition}

We call a key evolving PKE scheme \emph{forward secure} if it fulfills the following definition of fs-CPA:

\begin{definition}[fs-CPA~\cite{canetti2003}]
    A key-evolving public key encryption scheme is forward secure if all PPT adversaries have only a negligible advantage in the following game:

    \begin{itemize}
        \item
        The challenger uses \(\FsKeyGen\) to generate a pair \((\PubKey, \PrivKey_0)\).

        \item
        The adversary issues one \(\FsChallenge(j, m_0, m_1)\) query.
        The challenger chooses \(b \unidraw \bits\) randomly and returns \(C \coloneqq \FsEncrypt(m_b, j, \PubKey)\).
        
        \item
        The adversary issues one \(\FsBreakin(i)\) query.
        The challenger replies by using \(\FsUpdate\) \(i\) times to generate \(\PrivKey_i\) and returns it.

        \item
        The queries are subject to \(0 \le j < i\).

        \item
        The adversary now guesses \(b'\) and wins if \(b' = b\).
    \end{itemize}
\end{definition}

We use the name \emph{ratchet} to refer to a fs-CPA PKE.

\section{Models \& Goals}

In this section, we introduce how we model our system and which goals we want to achieve.

\subsection{System Model}

In our system, we assume that there is no trusted third party or a shared setup phase.
At any point in time, users who were previously not part of the system can join and participate in the network.
The only thing that we require is that a user can retrieve the public key of the person they want to communicate with out-of-band, for example by retrieving the key over the internet in advance, or by sharing it locally via QR codes.

We model our communication by using broadcast based wireless media.
We assume that each device can broadcast to any device within a certain radius, and we assume that there is a limit to the bandwidth with which the device can transmit data.
Therefore, each transmission takes a non-zero amount of time to complete

We do however \emph{not} make any assumption about interference due to neighboring devices.
We assume that enough free frequencies are available for the devices to use, and that the underlying technology stack will handle the cases where interference happens.
Existing apps have shown this to be a reasonable assumption:

Even at the highest activity levels, about 20 simultaneous Bluetooth peer-to-peer communications can succeed without interference~\cite{bocker2017}.
If we assume that not all devices exchange messages at the same time, but rather spread out their synchronizations, we can estimate the number of supported devices in a single Bluetooth coverage area.
With 30 synchronization runs per 30 seconds (1 second per synchronization), we can already support 600 devices (in a single Bluetooth range).
By reducing the Bluetooth transmission power and relying on the multi-hop message delivery, we can reduce the inference even further.

\subsection{Threat Model}

We assume the existence of a global, active adversary.
This adversary can eavesdrop on all traffic on the network as well as send own packets.

For the receiver, we assume that it will act honest up to a certain point in time \(T^\prime\).
Before \(T^\prime\), the receiver will not share any keys with the adversary and it will execute the protocol according to its specification.
After \(T^\prime\), the receiver will share its keys with the adversary.

We do consider attacks on the availability of the network (such as via DoS or jamming) to be out of scope for this work.

\subsection{Goals}

Our goal is to provide \emph{forward-secret unidirectional anonymous communication} on the application layer against a global active adversary.

\paragraph{Anonymous Communication}
Specifically, we want to achieve receiver-message unlinkability, which ensures that the adversary cannot tell who received a message or for whom it is intended.
The corresponding notion \(\RML\) is defined as a generic indistinguishability game in \cite{kuhn2019} as follows:

\begin{definition}[Kuhn's Privacy Game~\cite{kuhn2019}]
    The privacy game is defined between a challenger \(\challenger{}\), an adversary \(\adversary{}\) and the protocol \(\protocol{}\):

    \begin{enumerate}
        \item
        \(\challenger{}\) picks a challenge bit \(b\).

        \item
        \(\adversary{}\) sends a query containing two scenarios \(r_0\) and \(r_1\) containing transcripts of conversations to \(\challenger{}\).

        \item 
        \(\challenger{}\) verifies that the query is valid, i.e. that the scenarios only differ in the information that is supposed to be protected according to the notion \(\notion{}\) in question.

        \item 
        If the query is valid, \(\challenger{}\) forwards \(r_b\) to \(\protocol{}\).

        \item 
        The output \(\protocol{}(r_b)\) of \(\protocol{}\) is given to \(\adversary{}\).

        \item 
        \(\adversary{}\) gives a guess \(b'\).
    \end{enumerate}

    We say that a protocol \(\protocol{}\) fulfills notion \(\notion{}\) when \(b' = b\) only with negligible advantage.
\end{definition}

For \(\notion{} = \RML{}\), we require that neither scenario contains the empty communication, that each scenario contains the same number of messages received for each recipient and that the scenarios differ only in the recipients of the messages.

For this paper, we focus on receivers, as sender-message-unlinkability is hard to achieve in a MANET-based system, assuming that a user can be \enquote{singled out} and surrounded:
If the adversary sees the user send a message that they did not previously receive, the adversary can assume that the user sent the message.
That is, even if the message \emph{theoretically} contains no information linking it to the sender, the physical act of sending the message is still visible.

\paragraph{Forward Secrecy}
Further, want the receiver-message unlinkability to be forward secret.
In the case of receiver corruption in epoch \(T^\prime\), the message unlinkability still holds for all messages from epochs \(T\) before \(T^\prime\).
For the current epoch, receiver-message unlinkability trivially breaks for corrupted receivers, as they have to be able to decrypt incoming messages intended for them.

\paragraph{Unidirectionality}
We require a protocol that does not perform bidirectional handshakes or ratchet updates, as those are unreliable over MANETs.
Instead, conversations with new contacts should require only a key to be transferred, which can be done in a variety ways out-of-band, such as using QR codes, Bluetooth, or the internet (if available).

We do note that this version of the protocol does not hide communication patterns.
We leave it as future work to provide stronger notions than \(\RML\).

\section{Design}

To achieve our goals, our design must work without a reliable transport between two devices, provide receiver-message unlinkability, and be forward-secret.
As we cannot rely on the commonly used double ratchet or the internet for our application, we will discuss our new design in this section.

We design \prot{} to exchange messages using Bluetooth Low Energy.
Devices exchange messages with other devices in their range using a \emph{smart broadcast}~\cite{perry2022}, a technique in which a device requests only those messages from its neighbors that it has not yet seen.
This minimizes unnecessary transmissions of messages that have already arrived, but still allows messages to travel over multiple hops to reach their destination.

We do not employ any routing algorithm and instead rely on flooding to distribute messages.
This minimizes the information that an adversary can learn from the routing information or message path, and ensures that the message cannot be linked to its recipient.
We combine the flooding with \emph{implicit addressing}, which means that the recipient's address is not explicitly encoded in the message.
Instead, a device tries to decrypt every incoming message, and disregards those for which decryption fails.

For forward secrecy, we use ratchets based on time rather than ratchets based on message sequences.
This means that we avoid issues that are caused by missing or out-of-order messages.
To do this, we divide time into equally sized \emph{epochs}, each epoch being \(\Delta{} t\) seconds long.
Forward secrecy then guarantees that a key from epoch \(i\) does not leak information about ciphertexts from epochs \(i' < i\).

To ensure that devices are in the same epoch, we employ an algorithm to synchronize time in a MANET without a trusted time authority.
We note that slight deviations in the clock synchronization are tolerable, as they only affect the message delivery at the boundaries of an epoch.

Further, we avoid key negotiation issues by not relying on bidirectional communication between users.
Instead, we use a ratchet that is both \emph{non-interactive} and \emph{asymmetric}.
The non-interactiveness allows the devices to evolve their keys without communication, and the asymmetric property allows the user to easily add new contacts, without an interactive setup phase, e.g. by scanning a QR code.

\section{Specification}

In this section, we specify in more detail how we implement our design.

First, we describe the time synchronization, which ensures that all participants share approximately the same time so that their key ratchets are in the same epochs.

We then describe the forward-secure communication phase.
The security is based on a public key encryption scheme, so that the user doesn't need to agree on a key with every communication partner.
This also makes it easier to share keys via means such as QR codes.
For addressing, we use implicit addressing as in previous schemes~\cite{perry2022}.

\subsection{Time Synchronization}\label{sec:time-synchronization}

In practice, we can assume that the clocks are \enquote{roughly synchronized} up to a granularity of a few minutes --- users can manually synchronize the clocks using available reference points.
For shorter ratcheting intervals however, we need to synchronize the clocks more precisely, as otherwise the likelihood that a sender and receiver are not in the same epoch increases.

Due to our ad-hoc setting, we require a time synchronization algorithm that works in a network without a trusted leader, that can handle churn and that doesn't require reliable transport.
Byrenheid \etal{} propose an algorithm for root-node election via local voting in peer-to-peer networks~\cite{byrenheid2020},
which fits our network restrictions, and which we can adapt to do time synchronization:

In addition to its own clock, a device keeps \(n\) additional clocks, one for each neighbor that it knows about.
Every \(\Delta s\) seconds, a device then synchronizes its time with its neighbors by doing a \emph{three-majority vote}.
In this vote, the device picks three random clock values from the set of known neighborhood clock values, and if two of them are equal (within a small margin of error due to transmission delays), it will adopt this value as its new clock.
Otherwise, it picks a random clock out of the three and uses it.
After the vote, the device broadcasts its new time so that the neighboring devices can use this value in their votes, and it clears the list of known neighbor clocks to prepare for the next vote.

We note that a small timing error does not prevent communication, as long as the devices are still in the same epoch.
By increasing the duration of the epoch, the system can be made less prone to synchronization errors --- at the cost of increasing the period of time that an attacker can compromise with a leaked key.

We show pseudocode of the algorithm in \cref{alg:three-majority-vote}.

\begin{algorithm}
    \begin{algorithmic}
        \State \(i, j, k \unidraw{} N_u\)
        \If{\(\abs{t_i - t_j} < \varepsilon \vee \abs{t_i - t_k} < \varepsilon\)}
            \State \(t \coloneqq t_i\)
        \ElsIf{\(\abs{t_j - t_k} < \varepsilon\)}
            \State \(t_u \coloneqq t_j\)
        \Else
            \State \(t_u \unidraw{} \left\{t_i, t_j, t_k\right\}\)
        \EndIf
        \State \(N_u \leftarrow \emptyset\)
    \end{algorithmic}
    \caption{
        Three-Majority-Vote algorithm as executed by user \(u\).
        \(N_u\) denotes the set of \(u\)'s direct neighbors.
        \(t_x\) denotes user \(x\)'s clock value.
    }\label{alg:three-majority-vote}
\end{algorithm}

In addition to the initial synchronization of the clocks, we need to account for the fact that different clocks do not advance at the same speed.
Therefore, synchronized clocks may drift apart again over time.
We can expect this drift to be in the order of one millisecond per minute~\cite{akhmetyanov2021}.
As such, we consider this drift to be negligible for our application.
If the network persists long enough for the drift to accumulate too much, the clock synchronization algorithm will re-synchronize the clocks.

We do note that in practical implementations, GPS could be used to synchronize clocks with high accuracy~\cite{lewandowski1991}.
However, this would constitute a trusted third party, and as such, is unacceptable in our model.

\subsection{Asymmetric Ratchet}\label{sec:asymmetric-ratchet}

For our forward-secure public-key encryption scheme, we use Canetti's transformation~\cite{canetti2003} to build an \emph{asymmetric ratchet} out of a hierarchical identity-based encryption scheme.
This allows any user to encrypt messages for the recipient just by having their public key, without further key negotiations.

A HIBE is a public key encryption scheme which arranges users into a hierarchical tree of identities, and allows an identity's public and secret key to be derived from the parent's public and secret key, respectively:

\begin{definition}[HIBE]\label{def:hibe}
    A hierarchical identity-based encryption scheme (HIBE) consists of five algorithms \((\HibeKeyGen{},\allowbreak \HibeUskGen{},\allowbreak \HibeEnc{},\allowbreak \HibeDel{},\allowbreak \HibeDec{})\)~\cite{blazy2014}:
    
    \begin{itemize}
        \item
        \(\HibeKeyGen{}(\lambda) \rightarrow (\PubKey, \MasterSecret)\) is the key generation algorithm that outputs the public key and the master secret key.
    
        \item
        \(\HibeUskGen{}(\MasterSecret, \HibeId) \rightarrow (\PrivKey_\HibeId, \DelKey_\HibeId)\) takes the master secret key and an identity, and outputs the identity's secret key and delegation key.
    
        \item
        \(\HibeDel{}(\PrivKey_\HibeId, \DelKey_\HibeId, \HibeId') \rightarrow (\PrivKey_{\HibeId'}, \DelKey_{\HibeId'})\) takes the secret key and delegation key of an identity \(\HibeId\), and can generate the secret key and delegation key of an identity \(\HibeId'\) that is a descendant of \(\HibeId\).
    
        \item
        \(\HibeEnc{}(\PubKey, \HibeId) \rightarrow (K, C)\) takes the public key and an identity, and outputs an key and its encapsulation.
    
        \item
        \(\HibeDec{}(\PrivKey_\HibeId, \HibeId, C) \rightarrow K\) takes the secret key, the identity and a ciphertext, and returns the decapsulated key.
    \end{itemize}
\end{definition}

The transformation then makes use of the hierarchical delegation feature of the HIBE to provide key evolution:
Each user has their own instance of the HIBE, such that the master public key becomes the user's public key.
The master secret key is used to derive the secret key of the root identity, but is then deleted.

Now we treat the \enquote{root identity} as the starting point, and with each key evolution step, we move along the tree of identities in a depth-first traversal --- treating each identity as an epoch.
As we move down the tree, we generate the keys for all child nodes and store them for later use.
This way, we can delete the parent key, ensuring that we cannot derive the keys to previous epochs in the future.
When we reach the maximum tree depth, we traverse to the sibling node using the previously stored keys.

To encrypt something for the user, we use \(\HibeEnc\), providing the user's public key (which is the HIBE master public key) and the identity that represents the current epoch.

This scheme can be implemented efficiently by using a stack to store the generated secret keys, with the current key being on top of the stack.
We show pseudocode that implements the transformation in terms of the HIBE functions in \cref{alg:fskeygen,alg:fsupdate,alg:fsencrypt,alg:fsdecrypt}.
Additionally, we show the concept of identities-as-epochs in \cref{fig:identity-epoch-tree}.

\begin{algorithm}
    \begin{algorithmic}
        \State \((pk, sk) \leftarrow \HibeKeyGen{}\)
        \State \(t \coloneqq 0\)
        \State \((usk, udk) \leftarrow \HibeUskGen{}(t, sk)\)
        \State \Return public key: \((t, pk)\)
        \State \Return secret key: \((t, [(usk, udk)])\)
    \end{algorithmic}
    \caption{\(\FsKeyGen\)}\label{alg:fskeygen}
\end{algorithm}

\begin{algorithm}
    \begin{algorithmic}
        \State \Input \((t, [(usk_1, udk_1), \ldots, (usk_i, udk_i)])\)
        \State \(t' \coloneqq \mathrm{next}(t)\)
        \If{\(\mathrm{leaf}(t)\)}
            \State \Return \((t', [(usk_1, udk_1), \ldots, (usk_{i-1}, udk_{i-1})])\)
        \Else
            \State \((usk_\rightarrow, udk_\rightarrow) \leftarrow \HibeDel{}(usk_i, udk_i, t, \mathrm{right}(t))\)
            \State \((usk_\leftarrow, udk_\leftarrow) \leftarrow \HibeDel{}(usk_i, udk_i, t, \mathrm{left}(t))\)
            \State \Return \((t', [(usk_1, udk_1), \ldots, (usk_\rightarrow, udk_\rightarrow),\allowbreak (usk_\leftarrow,\allowbreak udk_\leftarrow)])\)
        \EndIf
    \end{algorithmic}
    \caption{\(\FsUpdate\)}\label{alg:fsupdate}
\end{algorithm}

\begin{algorithm}
    \begin{algorithmic}
        \State \Input \((t, pk, m)\)
        \State \((K, C) \leftarrow \HibeEnc{}(pk, t)\)
        \State \(C' \leftarrow \mathrm{AES.Encrypt}(h(K), m)\)
        \State \Return \((C, C')\)
    \end{algorithmic}
    \caption{\(\FsEncrypt\)}\label{alg:fsencrypt}
\end{algorithm}

\begin{algorithm}
    \begin{algorithmic}
        \State \Input \((t, [\ldots, (usk_i, udk_i)], (C, C'))\)
        \State \(K \leftarrow \HibeDec{}(usk_i, t, C)\)
        \State \(m \leftarrow \mathrm{AES.Decrypt}(h(K), C')\)
        \State \Return m
    \end{algorithmic}
    \caption{\(\FsDecrypt\)}\label{alg:fsdecrypt}
\end{algorithm}

\definecolor{colorFirstDerive}{HTML}{1CA216}
\definecolor{colorSecondDerive}{HTML}{9C16A2}
\begin{figure}
    \centering
    \begin{tikzpicture}[
        epoch/.style={draw,rectangle,rounded corners=1mm},
        firstDerive/.style={color=colorFirstDerive},
        secondDerive/.style={color=colorSecondDerive},
    ]
        \node[epoch] (n0) at (0, 0) { 0 };
        \node[epoch,firstDerive] (n1) at (-0.7, -1) { 1 };
        \node[epoch,secondDerive] (n2) at (-1, -2) { 2 };
        \node[epoch,secondDerive] (n3) at (-0.4, -2) { 3 };
        \node[epoch,firstDerive] (n4) at (0.7, -1) { 4 };
        \node[epoch] (n5) at (0.4, -2) { 5 };
        \node[epoch] (n6) at (1, -2) { 6 };

        \draw (n0) -- (n1);
        \draw (n0) -- (n4);
        \draw (n1) -- (n2);
        \draw (n1) -- (n3);
        \draw (n4) -- (n5);
        \draw (n4) -- (n6);

        \draw[->, firstDerive, dashed] (n0) edge[bend right] (n1);
        \draw[->, secondDerive, dashed] (n1) edge[bend right] (n2);

        \draw[dotted,firstDerive] (n1) -- (n4);
        \draw[dotted,secondDerive] (n2) -- (n3);
    \end{tikzpicture}
    \caption{Identities labelled with their epoch number. A key in the tree can be derived from the key of its parent. When switching from epoch \textcolor{colorFirstDerive}{0 to 1}, the key for epoch 4 is also generated and stored for future reference so that key 0 can be deleted. When switching from \textcolor{colorSecondDerive}{1 to 2}, the key for epoch 3 is stored and key 1 is deleted.}
    \label{fig:identity-epoch-tree}
\end{figure}
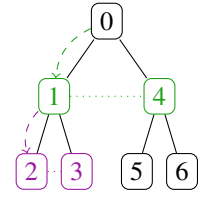

\subsection{The \prot{} Protocol}

We now combine the previously specified elements to create the \prot{} protocol.
In our protocol, each device runs the time synchronization (\cref{sec:time-synchronization}) in the background to keep its epoch counter synchronized with nearby devices by doing periodic three-majority votes.
When a user then sends a message \(m\), it is encrypted with the recipient's public key and the current epoch \(\xi \leftarrow \FsEncrypt(t, pk_{\mathrm{recipient}}, m)\) before the encrypted payload \(\xi\) is sent to the devices in reach.

When a device receives an encrypted payload, it floods that payload to its neighboring devices.
In addition, the device attempts to decrypt it using its secret key and the current epoch: \(m \leftarrow \FsDecrypt(t, sk, \xi)\).
If the decryption succeeds, the message is displayed to the user.
An overview of the protocol is shown in \cref{fig:overview}.

\newcommand\timeSync[1]{%
    \draw[->,gray] (left-na) -- node[below=0.5mm,inner sep=0pt,midway,sloped,gray] { \scriptsize\(\text{\faClockO} = 8\)} (#1)%
}
\newcommand\msgSend[1]{%
    \draw[->,gray] (right-na) -- node[midway,sloped,fill=white,inner sep=0pt] { \scriptsize\faEnvelope\textsubscript{\,\faLock} } (#1)
}
\newcommand\msgForward[2]{
    \draw[->,gray] (#1) -- node[midway,sloped,fill=white,inner sep=0pt] { \scriptsize\faEnvelope\rlap{\textsuperscript{\,\faMailForward}}\textsubscript{\,\faLock} } (#2)
}
\definecolor{colorTimeSync}{HTML}{1CA216}
\definecolor{colorMessageFlood}{HTML}{9C16A2}
\begin{figure}
    \centering
    \begin{tikzpicture}[
        device/.style={draw,circle,inner sep=0pt,minimum size=4mm},
        timeSync/.style={color=colorTimeSync},
        messageFlood/.style={color=colorMessageFlood},
    ]
        \draw[timeSync,rounded corners=1mm] (-3.9, 4) rectangle (-0.1, 0);
        \draw[messageFlood,rounded corners=1mm] (3.9, 4) rectangle (0.1, 0);

        \node[timeSync,align=center] at (-2, 4.2) { \textsc{Time Sync} };
        \node[messageFlood,align=center] at (2, 4.2) { \textsc{Message Flood} };


        \node[timeSync,device] (left-na) at (-2, 1) {a};
        \node[timeSync,device] (left-nb) at (-1, 1.5) {b};
        \node[timeSync,device] (left-nc) at (-3, 0.7) {c};
        \node[timeSync,device] (left-nd) at (-2.6, 2) {d};
        \node[timeSync,device] (left-ne) at (-3, 3) {e};
        \node[timeSync,device] (left-nf) at (-1.6, 3.3) {f};

        \timeSync{left-nb};
        \timeSync{left-nc};
        \timeSync{left-nd};

        \draw[gray, dashed] (-1.9, 1.9) rectangle (-0.7, 2.9);
        \node[align=left,gray,execute at begin node=\setlength{\baselineskip}{5pt}] at (-1.3, 2.35) {%
            \scriptsize\(\text{\faClockO}_a = 8\)\\
            \scriptsize\(\text{\faClockO}_e = \ldots\)\\
            \scriptsize\(\text{\faClockO}_f = \ldots\)\\
        };
        \draw[gray, dashed] (-1.9, 2.1) edge[bend right] (left-nd);
        

        \node[messageFlood,device] (right-na) at (2, 1) {a};
        \node[messageFlood,device] (right-nb) at (-1+4, 1.5) {b};
        \node[messageFlood,device] (right-nc) at (-3+4, 0.7) {c};
        \node[messageFlood,device] (right-nd) at (-2.6+4, 2) {d};
        \node[messageFlood,device] (right-ne) at (-3+4, 3) {e};
        \node[messageFlood,device] (right-nf) at (-1.6+4, 3.3) {f};

        \msgSend{right-nb};
        \msgSend{right-nc};
        \msgSend{right-nd};

        \msgForward{right-nd}{right-ne};
        \msgForward{right-nd}{right-nf};
    \end{tikzpicture}
    \caption{Overview of the two protocol components. On the left, the \textcolor{colorTimeSync}{time synchronization} is shown, in which device \enquote{a} broadcasts its current epoch and device \enquote{d} keeps track of the epoch broadcasts it has seen. On the right, the \textcolor{colorMessageFlood}{message flooding} is shown, in which \enquote{a} sends a message that is flooded all the way to \enquote{f}.}\label{fig:overview}
\end{figure}
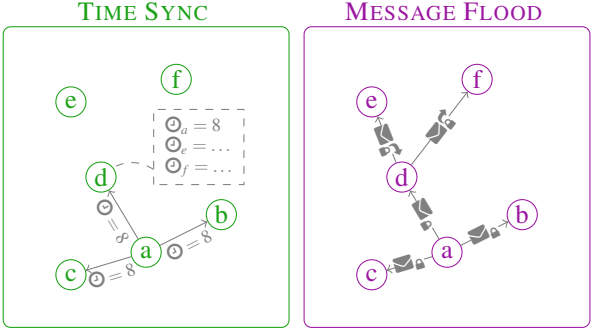
\let\timeSync\undefined
\let\msgSend\undefined
\let\msgForward\undefined

\subsection{Implementation}

We have implemented our protocol as an Android application, using Bluetooth and Bluetooth Low Energy to transmit messages.
The necessary cryptographic primitives are implemented in Rust and provided as Android NDK libraries, while the user interface and wireless transfer are implemented in Kotlin.

As the basis for our HIBE, we choose the construction by Blazy \etal{}~\cite{blazy2014}.
This produces little overhead and ensures that the ciphertexts cannot be linked to a key.
We have implemented this scheme on top of the BLS12-381 curve.

To improve the successful transport of messages close to the end of an epoch, we also implement a technique we call \emph{smooth epoch rollover}:
We only use half of the desired epoch time \(\Delta{}t\), but keep the key of the previous epoch in memory.
This way, a message that is sent close before the end of an epoch can still be decrypted, but the overall forward secrecy does not decrease.

\section{Security}

The security of our system comes in two parts.
The first part is the semantic security of ciphertexts in the sense that an adversary cannot learn information from encrypted messages.
The second part is the anonymity of ciphertexts in the sense that an adversary cannot tell for whom a ciphertext is destined.
As our scheme is forward secret, we want both securities to hold even if the adversary corrupts a future key.

\subsection{Semantic Security}

The semantic security stems from the security of the underlying HIBE.
Canetti's transformation ensures that the resulting PKE is fs-CPA secure if the HIBE is secure~\cite[Theorem 4]{canetti2003}.
As the underlying HIBE is SN-CPA secure, our scheme is fs-CPA secure.

\subsection{Receiver-Message-Unlinkability}

In order to prove that \prot{} fulfills \(\RML\), we need to prove that an adversary cannot link ciphertexts to a specific person based on their public key.
This property is usually called key-privacy.
In our model, key-privacy corresponds to receiver-message-unlinkability, as there are no other pieces of information that could link a message to its recipient (i.e. there is no routing information or directed sending).

While our underlying HIBE does promise anonymity, it does so in the context of a single instantiation of the scheme.
The anonymity therefore refers to the HIBE identity, which in our case is used as the epoch identifier.
Translated to our context, the scheme is \enquote{epoch anonymous}.

What we need to show in addition is that ciphertexts from different instantiations cannot be distinguished, which then means that ciphertexts cannot be linked to public keys and therefore to users in our system.
This gives us proper key-privacy for our forward secret encryption scheme.

We formalize this using the FS-ANON game as follows:

\begin{definition}{FS-ANON}
    Given an adversary \(\adversary{}\) and a challenger \(\challenger{}\).
    The FS-ANON game proceeds as follows:
    
    \begin{enumerate}
        \item
        \(\challenger{}\) generates two key pairs, \((\PubKey{}_0, \PrivKey{}_0) \leftarrow \FsKeyGen{}\) and \((\PubKey{}_1, \PrivKey{}_1) \leftarrow \FsKeyGen{}\).
    
        \item
        \(\challenger{}\) passes \(\PubKey{}_0\), \(\PubKey{}_1\) to \(\adversary{}\).
    
        \item
        \(\adversary{}\) passes a message \(m\) to \(\challenger{}\).
    
        \item
        \(\challenger{}\) randomly chooses \(b \unidraw{} \bits{}\).
    
        \item
        \(\challenger{}\) computes \(c \leftarrow \FsEncrypt{}(0, pk_b, m)\) and sends \(c\) to \(\adversary{}\).

        \item
        \(\challenger{}\) computes \(\PrivKey{}'_0 \leftarrow \FsUpdate{}(\PrivKey{}_0)\), \(\PrivKey{}'_1 \leftarrow \FsUpdate{}(\PrivKey{}'_1)\) and sends them to \(\adversary{}\).
    
        \item
        \(\adversary{}\) produces a guess \(b'\) and sends \(b'\) to \(\challenger{}\).
    
        \item
        \(\adversary{}\) wins if \(b = b'\).
    \end{enumerate}
    
    We call a scheme FS-ANON secret if the winning advantage of any PPT adversary \(\adversary{}\) is negligible.
\end{definition}

We can see that the game captures two properties:
The adversary has access to both public keys, therefore the game captures the notion that the adversary can link ciphertexts to specific public keys.
Additionally, the adversary gets access to future secret keys, therefore the game captures the unlinkability in a \emph{forward secret} way.

We prove that our construction is FS-ANON secret using a hybrid argument:

\begin{game}\label{game:0}
    \Cref{game:0} is the original FS-ANON game.
\end{game}

\begin{game}\label{game:1}
    \Cref{game:1} is the same as \Cref{game:0}, except that in step 5 we change the computation of the ciphertext:
    \begin{align*}
        (K, C) &\unidraw{} \mathcal{K}^\mathrm{FS} \times \mathcal{C}^\mathrm{FS} \\
        C'     &\leftarrow \mathrm{AES.Encrypt}(h(K), m) \\
        c      &\coloneqq (C, C')
    \end{align*}
    That is, instead of properly encapsulating a key using \(\HibeEnc{}\), we now use a random key.
\end{game}

\begin{lemma}\label{lemma:1}
    The winning advantage of \(\adversary{}\) in \Cref{game:1} is negligible.
\end{lemma}

\begin{proof}
    The key used in \Cref{game:1} is random and does not depend on either of the selected public keys.
    As such, the challenge ciphertext \(c\) contains no information about the key that has been used.
\end{proof}

\begin{lemma}\label{lemma:0-1}
    No PPT adversary can distinguish \Cref{game:0} and \Cref{game:1}.
\end{lemma}

\begin{proof}
    Assuming that there is an adversary \(\adversary{}\) that distinguishes \Cref{game:0} and \Cref{game:1}, we will build an adversary \(\adversary{}'\) that breaks PR-HID-CPA.

    We note that the space \(\mathcal{K}^\mathrm{FS}\) is \(\Gt\), and thus coincides with the key space \(\mathcal{K}\) of the underlying HIBE.
    The space \(\mathcal{C}^\mathrm{FS}\) is \(\Ge^{k + 1} \times \Ge^{n}\), and thus coincides with the ciphertext space \(\mathcal{C}\) of the underlying HIBE.
    Both of those spaces only depend on the public group parameters, and do \emph{not} depend on any elements of the master key.

    We now construct \(\adversary{}'\) to break PR-HID-CPA with the help of \(\adversary{}\):

    \begin{enumerate}
        \item
        \(\adversary{}'\) calls \textsc{Initialize} to get \((\mathrm{pk}_0, \mathrm{dk}_0)\).

        \item
        \(\adversary{}'\) calls \(\textsc{USKGen}\) to request the secret keys and build \(\PrivKey{}'_0\) as shown in \Cref{alg:fsupdate}.

        \item
        \(\adversary{}'\) calls \(\HibeKeyGen{}\) to get \((\mathrm{pk}_1, \mathrm{sk}_1, \mathrm{dk}_1)\).

        \item
        \(\adversary{}'\) uses \(\HibeDel\) to build \(\PrivKey{}'_1\) as shown in \Cref{alg:fsupdate}.

        \item 
        \(\adversary{}'\) sends \((0, [(\mathrm{pk}_0, \mathrm{dk}_0)]\) and \((0, [(\mathrm{pk}_1, \mathrm{dk}_1)])\) to \(\adversary{}\).

        \item
        \(\adversary{}\) sends the message \(m\) to \(\adversary{}'\).
        
        \item
        \(\adversary{}'\) calls \(\textsc{Enc}(\mathrm{root})\), where \(\mathrm{root}\) is the root of the identity hierarchy, to retrieve the challenge tuple \((K^*, C^*)\).

        \item
        \(\adversary{}'\) computes \(c \coloneqq \mathrm{AES.Encrypt}(h(K^*), m)\) and sends \((C^*, c)\) to \(\adversary{}\).

        \item
        \(\adversary{}\) sends the guess \(b'\) to \(\adversary{}'\), who then sends it to the PR-HID-CPA challenger.
    \end{enumerate}

    Since Canetti's transformation requires Step 2 to only requests the keys of identitites that are either descendant (suffixes) of the current epoch, or siblings, a prefix is never requested.
    Hence the post-condition of the PR-HID-CPA game is never violated.
\end{proof}

\begin{theorem}
    Given a HIBE that is PR-HID-CPA~\cite{blazy2014} secure, our construction is FS-ANON secret.
\end{theorem}

\begin{proof}
    The proof follows from \Cref{lemma:1} and \Cref{lemma:0-1}.
\end{proof}

\subsection{Robustness of the Time Synchronization}

An adversary can affect the time synchronization either by introducing many malicious devices with false time broadcasts, or by deviating from the protocol altogether (e.g., by sending different times to each of the neighboring devices).
However, such attacks only reduce the availability of the system, 
and are therefore out of scope for our purposes.

\section{Evaluation}

We perform our evaluation in multiple steps:
First, we perform microbenchmarks to evaluate the usability of our cryptographic primitives in practice.
Then, we use the \omnet{} simulation framework\footnote{\url{https://omnetpp.org}} to evaluate the time synchronization and the ratio of successfully arrived messages in different scenarios.
Finally, we use our Android prototype perform a practical test.

\subsection{Encryption Microbenchmark}

As a first step, we do microbenchmarks to evaluate whether our encryption scheme is fast enough on mobile devices to keep up with the rate of incoming messages.
For that, we have used our implementation of the ratchet in Rust\footnote{\url{https://www.rust-lang.org/}}, and cross-compiled it to Android smartphones.
We also ran the benchmarks on a laptop\footnote{Lenovo Thinkpad E14 AMD G4, Ryzen 5 5625U, 16 GiB RAM} to provide reference values using the \texttt{criterion} tool\footnote{\url{https://crates.io/crates/criterion}}.

We expect the key generation and private key ratcheting algorithms to be the slowest, as they have to do the most work internally.
The public key ratcheting step only needs to adjust the HIBE identity, so we expect it to take negligible time.
Additionally, we expect decryption and encryption to be fast regardless of the data size, since they use a symmetric cipher for the data and only use the HIBE to derive a symmetric key.

Our experiments confirm our expectations:
Both encryption and decryption take less than 5 milliseconds (on a laptop) and 7 milliseconds (on a smartphone) for all tested message sizes from 512 bytes to 10 KiB.
Public key ratcheting takes almost no time, as the algorithm simply needs to advance the epoch counter.
Private key ratcheting takes around 12 milliseconds on a laptop, and around 121 milliseconds on a smartphone.
While this is a relatively long time, a user typically only has a single private key to ratchet, so in practice this is not a problem.

The slow key generation is negligible, as it is a one-time setup operation that can be done at installation time.

We also summarize those results in \cref{tab:runtime-ratcheting}.

\begin{table}
    \centering
    \begin{tabular}{ccc}
        \toprule
                               & ThinkPad        & Pixel 6 \\
        \midrule
        Key generation         & \qty{160.49}{\ms}   & \qty{259.09}{\ms} \\
        Public key ratcheting  & \qty{2.73}{\ns}     & \qty{79.45}{\ns} \\
        Private key ratcheting & \qty{11.52}{\ms}    & \qty{120.72}{\ms} \\
        Encryption             & \qty{3.93}{\ms}     & \qty{5.22}{\ms} \\
        Decryption             & \qty{1.76}{\ms}     & \qty{6.70}{\ms} \\
        \bottomrule
    \end{tabular}
    \caption{Average run times for our implementation of the algorithms in \cref{sec:asymmetric-ratchet}}
    \label{tab:runtime-ratcheting}
\end{table}

While our algorithms are slower than symmetric encryption or simpler public-key encryption methods, we can say that they are still fast enough to not be a problem in practice.

\subsection{Simulation Assumptions}\label{subsec:evaluation-assumptions}

Since the behavior of the network depends on the behavior of its participants, we need to specify our assumptions about the movement of the users and the chosen parameters.
We start by defining three movement models, which will form the basis of our evaluation:

\begin{itemize}
    \item
    \emph{Static}:
    In this model, all actors stand in a grid with no movement.
    Each device is connected to the devices of nearby users within a given radius.
    This model approximates the situation e.g. at a sit-in or sit-down strike, after the participants have taken their places and are no longer moving.

    \item
    \emph{Converging}:
    In this model, the actors form small groups that move and converge to a central point.
    Once they arrive at the central point, they will remain static for a while, before diverging again.
    This model describes the buildup to a large gathering, as multiple smaller groups converge on a central location.
    We include this model to evaluate if groups that are internally synchronized can \emph{merge} with other groups.

    \item
    \emph{Real-World}:
    In this model, we use motion data captured from a real festival~\cite{blanke2014}.
    While this data is not from a protest, we still use it to approximate real-life data from the gathering of a large number of people.
\end{itemize}

We have included the code that we used to generate the different models together with the code we use for the actual simulations.

We also set our system parameters to epochs of one minute, and a sending rate of one message every 30 seconds for the users.
The first message from a user is delayed by a random delay, to ensure that not all users send at the same time.
We also assume that the channels between the devices have a limited bandwith of 1.4\,Mbps.

For each parameter combination, we ran 16 simulation runs with different random starting seeds to produce a set of 16 result values.
Each simulation itself encompasses 5 hours of simulated time.

\subsection{Time Synchronization}

The next step is to evaluate the time synchronization aspect of our protocol.
Since our ratcheting is based on time, the clocks of the network participants must be synchronized, otherwise they will use keys from the wrong epoch.

To do this, we take the movement models described in \Cref{subsec:evaluation-assumptions} and add an \omnet{} simulation of the communication behavior which runs the three-majority time synchronization.
We then look at the time it takes for the network to reach a state where 90\% of its participants share the same time.
We also consider various parameter combinations to see how the system behaves as the number of users, the density of users or the number of attackers increases.

From Byrenheid \etal{}~\cite{byrenheid2020}, we can expect the synchronization to \emph{eventually} converge.
For a higher user count, we expect the synchronization to take longer, because there are more clocks to synchronize and more different clock values in the network.
At higher density, we expect the synchronization to be faster, as each user has more neighbors to communicate with.
Therefore, the clock values propagate through the network faster.
For a larger number of attackers, we expect the synchronization to be more chaotic, as the attackers can broadcast false values without adjusting their own clocks.

Our results show that the time synchronization can be as fast as 25 seconds in networks that are small or dense, but for larger networks or sparsely populated networks can take up to 10 minutes.
As attackers are introduced to the system, the synchronization becomes more chaotic, as the synchronization process depends on whether or not the adversarial times are part of the random selection in the three majority voting.

These results are shown in \Cref{fig:eval-tts}.
Additionally, when evaluating with real movement traces~\cite{blanke2014}, we can see that the convergence to a shared time is slower, but the algorithm still reaches 90\% of synchronized clocks after around 2 hours.

\begin{figure*}
    \centering
    \includegraphics[width=0.33\linewidth]{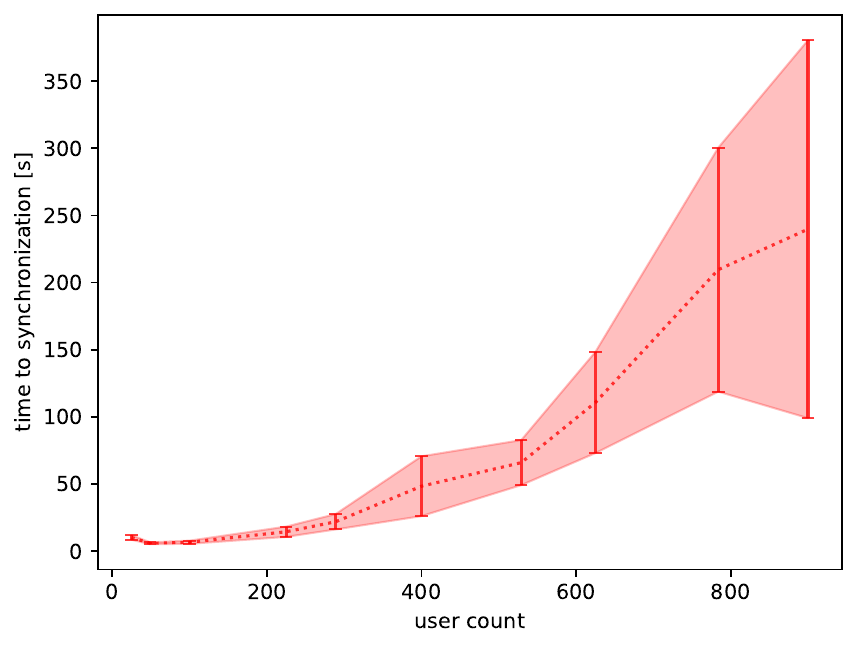}
    \includegraphics[width=0.33\linewidth]{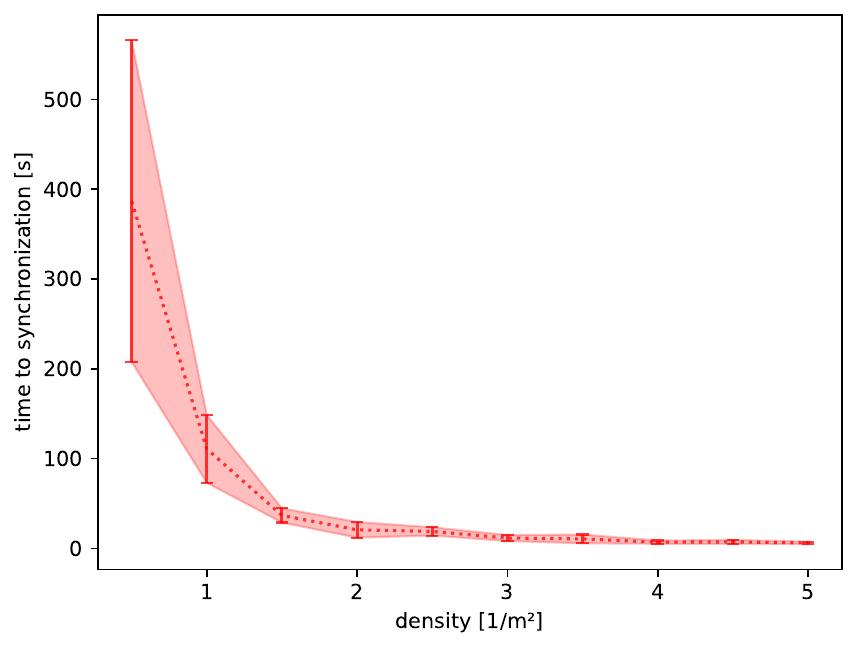}
    \includegraphics[width=0.33\linewidth]{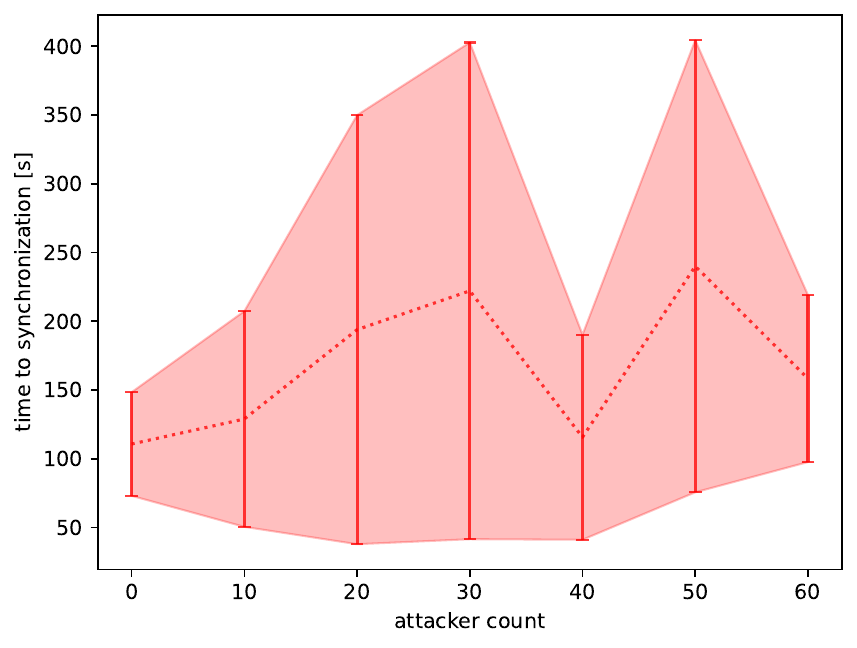}
    \caption{Time to network synchronization for a variable user count (left), a variable density (middle) and a variable attacker count (right). The standard parameters are 625 users, 1 user per square meter, and no attackers.}
    \label{fig:eval-tts}
\end{figure*}

When we look at the converging movement model, we can see that the time needed to synchronize the network is generally much higher, as the groups of users are initially out of range of each other.
However, synchronization is still achieved, confirming that our algorithm can merge multiple groups that are internally synchronized into one large group.
This result is shown in \Cref{fig:converging-tts}.

\begin{figure}
    \centering
    \includegraphics[width=\linewidth]{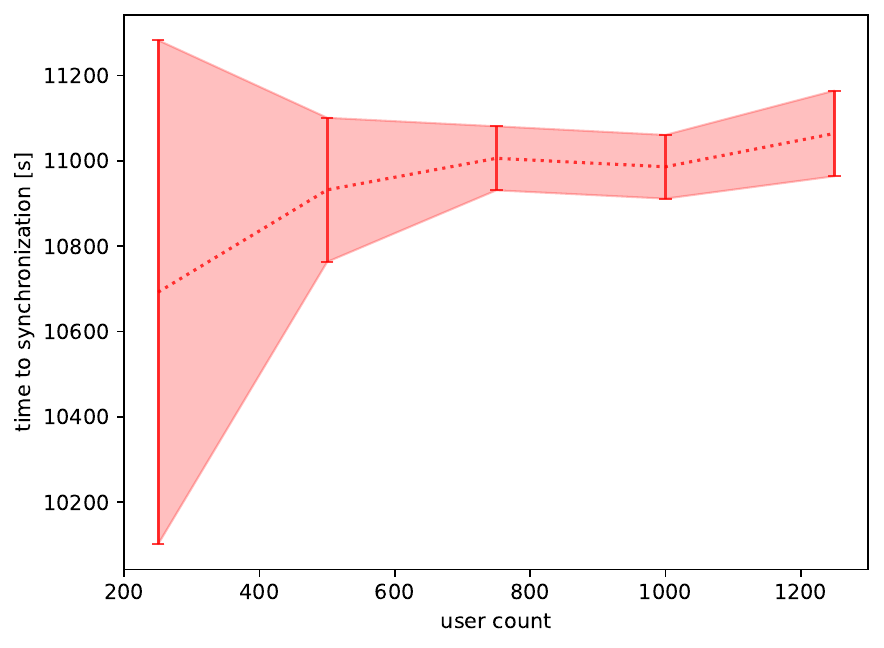}
    \caption{Time to network synchronization for a variable user count in the converging scenario.}
    \label{fig:converging-tts}
\end{figure}

For the real-world data, we can see that the synchronization takes a long time, as the users are initially not connected.
However, the share of people which share the same epoch continually increases, and it reaches 90\% after 2.5 hours.
This is shown in \Cref{fig:zueri-reachability}.

\begin{figure}
    \centering
    \includegraphics[width=\linewidth]{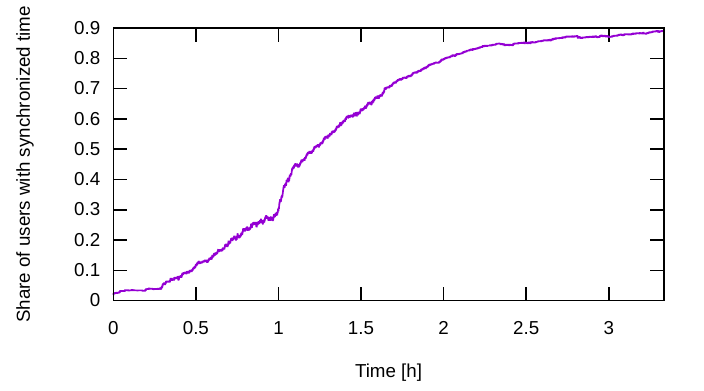}
    \caption{Share of the network in the real-life scenario that shares the same epoch as time progresses.}
    \label{fig:zueri-reachability}
\end{figure}

From this we can conclude that the time synchronization is good enough in practice to provide the required level of synchronization.

\subsection{Network Simulations}

We use our network model in \omnet{} to assess the overall message delivery in different scenarios.
To do so, we look at the average share of successfully delivered messages after an initial network stabilization period of 5 minutes.
A message is successfully delivered if it arrives at its intended recipient during the correct epoch (before the key is deleted).

For this experiment, we expect results similar to those for time synchronization:
For the user count, we expect only a small impact to the successful message ratio, up to the point where the bandwidth of the system is exhausted and more users cause messages to be dropped.
As we increase the density, we expect the successful message ratio to increase, as the better connectivity improves the performance of the system.
As we add attackers, we expect the ratio of successful messages to decrease, as the attackers disturb the time synchronization process.

Our results in \Cref{fig:eval-ar} show that the amount of attackers has a clear negative impact on the share of successful messages, dropping by 7 percentage points from 99\% to 92\%.
Additionally, increasing the number of users first improves the performance, but only up to a point of around 600 users.
After that, the negative effects of the higher bandwidth requirements and the longer time it takes to synchronize the clocks lead to a degradation in performance.
For the density, we can see that there is a slight improvement in the successful message ratio for denser crowds.

\begin{figure*}
    \centering
    \includegraphics[width=0.33\linewidth]{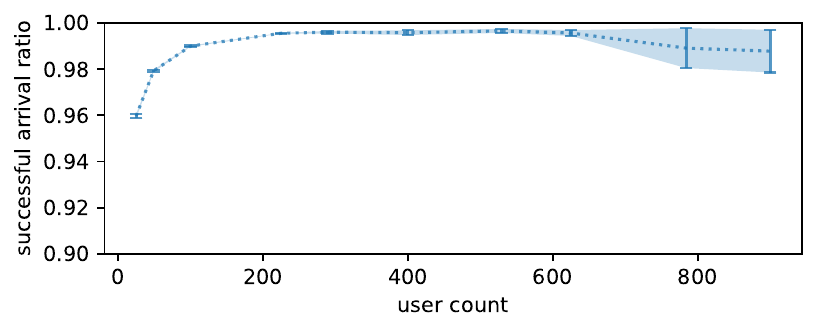}
    \includegraphics[width=0.33\linewidth]{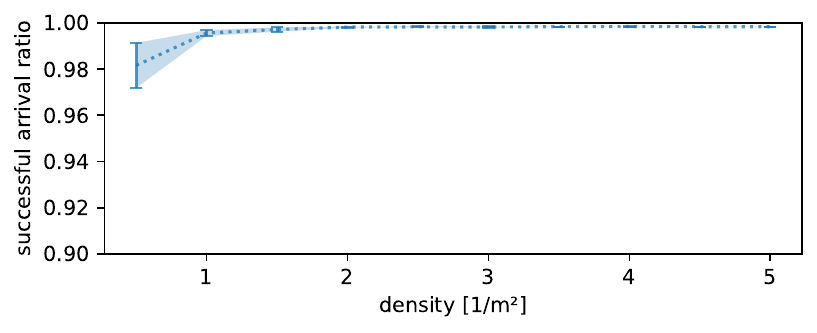}
    \includegraphics[width=0.33\linewidth]{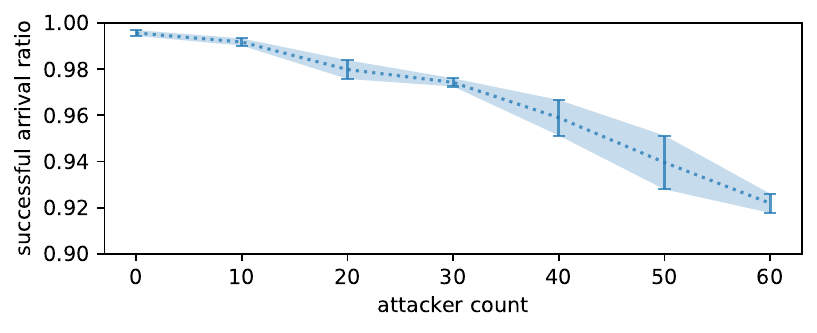}
    \caption{Average success ratio for a variable user count (left), variable density (middle) and variable attacker count (right). The standard parameters are 625 users, 1 user per square meter and no attackers.}
    \label{fig:eval-ar}
\end{figure*}

When we look at the success ratio in the converging model, we can see that fewer messages arrive successfully than in the static scenario:
Many users are out of range of each other during the phase in which the users converge, and during the phase in which the users diverge.
This result is shown in \Cref{fig:converging-ar}.

\begin{figure}
    \centering
    \includegraphics[width=\linewidth]{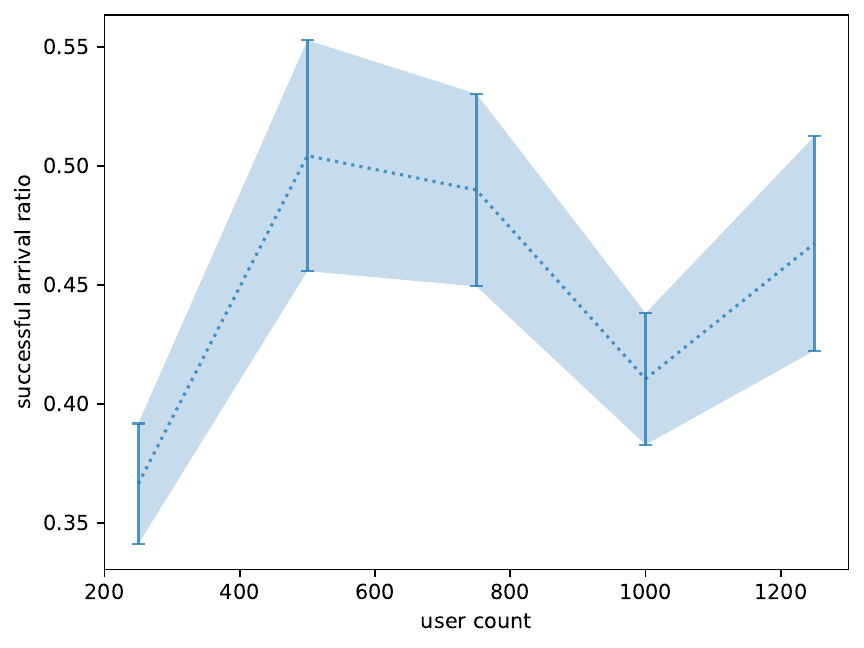}
    \caption{Average success ratio for a variable user count in the converging scenario.}
    \label{fig:converging-ar}
\end{figure}

When simulating on real-world movement, we can see that the initial share of successfully delivered messages is low, but then increases once the time synchronization has synchronized a larger portion of the network.
This result is shown in \Cref{fig:zueri-ar}.

\begin{figure}
    \centering
    \includegraphics[width=\linewidth]{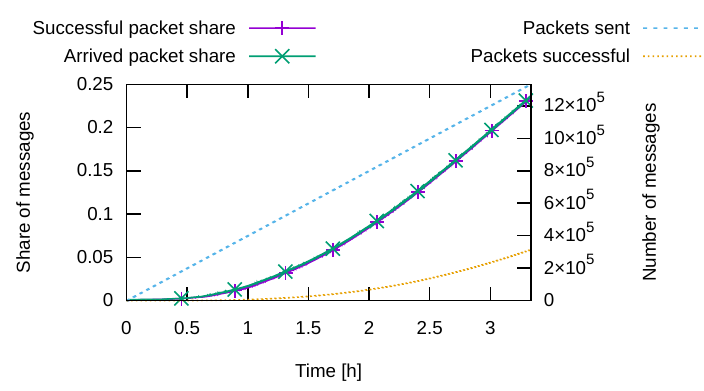}
    \caption{Messages in the real-life scenario that arrived (\emph{Arrived packet share}) and that arrived during the right epoch (\emph{Successful packet share}).}
    \label{fig:zueri-ar}
\end{figure}

\subsection{Proof-of-Concept Application}

We have used our proof-of-concept application to test our protocol on real hardware.
For that, we have installed the app on 4 of our devices with different processing power and sent messages over different distances in a crowded cafeteria.

Our results in \Cref{tab:poc-result} show that most messages arrive in the correct epoch and can be successfully decrypted.
Device 4 has lost the most messages, which was a result of the device being too far away from the remaining devices, such that it did not receive the messages in time.

\begin{table}
    \begin{tabular}{lrrrr}
    \toprule
                 & Addressed & Successful & Successful\% \\
    \midrule
        Dev. 1 & 23 & 17 & 74\% \\
        Dev. 2 & 5 & 4 & 80\% \\
        Dev. 3 & 4 & 3 & 75\% \\
        Dev. 4 & 16 & 7 & 44\% \\
    \bottomrule
    \end{tabular}
    \caption{Results from the proof-of-concept test.}\label{tab:poc-result}
\end{table}

\section{Conclusion}

We construct \prot{}, a protocol that provides forward secret messaging in unreliable ad-hoc networks without the need for bidirectional communication, improving on existing protocols that use standard constructions like the double-ratchet.
We formally prove that our protocol satisfies our forward-secrecy and key-privacy goals, resulting in a secure system.

Further, we show that our protocol is practical by providing benchmark results for the cryptographic operations as well as simulations for the network communication.
These simulations show that we can successfully deliver more than 90\% of the messages sent even in networks with 10\% adversarial users.

In addition to the functionality and security of our protocol, we are also improving the usability of existing solutions:
Users can add contacts by scanning their public key, without the need for a handshake.
This can be a great benefit in situations where access to the internet or other reliable communication channels has already been lost.

Finally, we provide a functional proof-of-concept implementation of our protocol as an Android app, that can be used and extended for further research.

\section*{Acknowledgements}
We would like to thank Dennis Hofheinz and Roman Langrehr for their cryptographic expertise, Ulf Blanke for the discussions around the evaluation on real-world data and Alexander Linder for his help in creating the prototype app.
This work was supported by the German Research Foundation (DFG, Deutsche Forschungsgemeinschaft) as part of Germany's Excellence Strategy -- EXC 2050/1 -- Project ID 390696704 Cluster of Excellence "Centre for Tactile Internet with Human-in-the-Loop" (CeTI) of Technische Universität Dresden and by 
the KASTEL Security Research Labs (Topic 46.23 of the Helmholtz Association).

\printbibliography

\end{document}